\documentclass{amsart}
\usepackage{pinlabel}
\usepackage{amsmath,  amsthm,  amssymb,  amscd,  mathrsfs, epsfig, color, float}
\usepackage[mathscr]{eucal}
\usepackage{hyperref}
\hypersetup{
    unicode=false,          
    pdftoolbar=true,        
    pdfmenubar=true,        
    pdffitwindow=false,     
    pdfstartview={FitH},    
    pdftitle={},    
    pdfauthor={Subhadip Chowdhury},     
    pdfsubject={},   
    pdfcreator={Subhadip Chowdhury},   
    pdfkeywords={} {}, 
    pdfnewwindow=true,      
    colorlinks=true,       
    linkcolor=red,          
    citecolor=green,        
    filecolor=magenta,      
    urlcolor=cyan           
}
\usepackage{todonotes}
\usepackage{enumerate}
\usepackage{pdfpages}
\usepackage{verbatim}
\usepackage{epstopdf}
\usepackage{caption}
\usepackage{subcaption}

\usepackage{tikz}
\usetikzlibrary{decorations.pathreplacing}
\usetikzlibrary{cd}
\usetikzlibrary{decorations.markings}

\graphicspath{{./2cells/}{./otherpics/}{./0and1cells/}}
\usepackage{array,longtable}

\usepackage{etoolbox}
\makeatletter
\def\@footnotecolor{blue}
\define@key{Hyp}{footnotecolor}{%
	\HyColor@HyperrefColor{#1}\@footnotecolor%
}
\patchcmd{\@footnotemark}{\hyper@linkstart{link}}{\hyper@linkstart{footnote}}{}{}
\makeatother
\hypersetup{footnotecolor=blue}

\usepackage{microtype}
\DisableLigatures{encoding = *, family = *}
\makeatletter
\let\save@mathaccent\mathaccent
\newcommand*\if@single[3]{%
  \setbox0\hbox{${\mathaccent"0362{#1}}^H$}%
  \setbox2\hbox{${\mathaccent"0362{\kern0pt#1}}^H$}%
  \ifdim\ht0=\ht2 #3\else #2\fi
  }
\newcommand*\rel@kern[1]{\kern#1\dimexpr\macc@kerna}
\newcommand*\widebar[1]{\@ifnextchar^{{\wide@bar{#1}{0}}}{\wide@bar{#1}{1}}}
\newcommand*\wide@bar[2]{\if@single{#1}{\wide@bar@{#1}{#2}{1}}{\wide@bar@{#1}{#2}{2}}}
\newcommand*\wide@bar@[3]{%
  \begingroup
  \def\mathaccent##1##2{%
    \let\mathaccent\save@mathaccent
    \if#32 \let\macc@nucleus\first@char \fi
    \setbox\z@\hbox{$\macc@style{\macc@nucleus}_{}$}%
    \setbox\tw@\hbox{$\macc@style{\macc@nucleus}{}_{}$}%
    \dimen@\wd\tw@
    \advance\dimen@-\wd\z@
    \divide\dimen@ 3
    \@tempdima\wd\tw@
    \advance\@tempdima-\scriptspace
    \divide\@tempdima 10
    \advance\dimen@-\@tempdima
    \ifdim\dimen@>\z@ \dimen@0pt\fi
    \rel@kern{0.6}\kern-\dimen@
    \if#31
      \overline{\rel@kern{-0.6}\kern\dimen@\macc@nucleus\rel@kern{0.4}\kern\dimen@}%
      \advance\dimen@0.4\dimexpr\macc@kerna
      \let\final@kern#2%
      \ifdim\dimen@<\z@ \let\final@kern1\fi
      \if\final@kern1 \kern-\dimen@\fi
    \else
      \overline{\rel@kern{-0.6}\kern\dimen@#1}%
    \fi
  }%
  \macc@depth\@ne
  \let\math@bgroup\@empty \let\math@egroup\macc@set@skewchar
  \mathsurround\z@ \frozen@everymath{\mathgroup\macc@group\relax}%
  \macc@set@skewchar\relax
  \let\mathaccentV\macc@nested@a
  \if#31
    \macc@nested@a\relax111{#1}%
  \else
    \def\gobble@till@marker##1\endmarker{}%
    \futurelet\first@char\gobble@till@marker#1\endmarker
    \ifcat\noexpand\first@char A\else
      \def\first@char{}%
    \fi
    \macc@nested@a\relax111{\first@char}%
  \fi
  \endgroup
}
\makeatother

\makeatletter
\renewcommand\section{\@startsection{section}{1}%
  \z@{.7\linespacing\@plus\linespacing}{.5\linespacing}%
  {\large\scshape\centering}}
\renewcommand\subsection{\@startsection{subsection}{2}%
  \z@{.5\linespacing\@plus.7\linespacing}{-.5em}%
  {\large\bfseries}}
\makeatother
\begin{document}

\title{A Topological proof that $O_2$ is $2$-MCFL}
\author{Subhadip Chowdhury}
\address{Department of Mathematics \\ University of Chicago \\
Chicago,  Illinois,  60637}
\email{subhadip@math.uchicago.edu}
\date{\today}

\begin{abstract}
We give a new proof of Salvati's theorem that the group language $O_2$ is $2$ multiple context free. Unlike Salvati's proof, our arguments do not use any idea specific to two-dimensions. This raises the possibility that the argument might generalize to $O_n$.
\end{abstract}

\maketitle
\setcounter{tocdepth}{1}
\tableofcontents


	\newtheorem{theorem}{Theorem}[section]
	\newtheorem{lemma}[theorem]{Lemma}
	\newtheorem{corollary}[theorem]{Corollary}
	\newtheorem{conjecture}[theorem]{Conjecture}
	\newtheorem{proposition}[theorem]{Proposition}
	\newtheorem{question}[theorem]{Question}
	\newtheorem{problem}[theorem]{Problem}
	\newtheorem*{thm_2D_derivable}{Theorem~\ref{thm:2D_derivable}}
	\newtheorem*{claim}{Claim}
	\newtheorem*{criterion}{Criterion}
	\theoremstyle{definition}
	\newtheorem{definition}[theorem]{Definition}
	\newtheorem{construction}[theorem]{Construction}
	\newtheorem{notation}[theorem]{Notation}
	\newtheorem{convention}[theorem]{Convention}
	\newtheorem*{warning}{Warning}

	\theoremstyle{remark}
	\newtheorem{remark}[theorem]{Remark}
	\newtheorem{example}[theorem]{Example}
	\newtheorem{scholium}[theorem]{Scholium}
	\newtheorem*{case}{Case}
	
	\newcommand\id{\textnormal{id}}
	\newcommand\Z{\mathbb Z}
	\newcommand\R{\mathbb R}
	\newcommand\RP{\mathbb RP}
	\newcommand\T{\mathbb T}
	\newcommand\D{\mathbb D}
	\newcommand\B{\mathbb B}
	\newcommand\X{\times}
	\newcommand\inv{\textnormal{Inv}}
	\newcommand{\perm}{\textnormal{perm}}
	\newcommand{\genus}{\textnormal{genus}}
	\newcommand{\homeo}{\textnormal{Homeo}}
	\newcommand{\rot}{\textnormal{rot}}
	\newcommand{\Hom}{\textnormal{Hom}}
	\newcommand{\SL}{\textnormal{SL}}
	\newcommand{\fr}{\textnormal{fr}}
	\newcommand{\vf}[1]{\frac{1}{#1}}
	\newcommand{\lar}{\leftarrow}
	\newcommand{\la}{\langle}
	\newcommand{\ra}{\rangle}
	\newcommand{\Circ}{\textnormal{Circ}}
	\renewcommand{\d}{\partial}
	\renewcommand{\i}{\iota}
	
	\newcommand{\mf}{\mathfrak}
	\newcommand{\mc}{\mathcal}
	\renewcommand{\d}{\partial}
	\renewcommand{\v}{\vec{v}}
	\newcommand{\Lk}{\textup{Lk}}
	\renewcommand{\phi}{\varphi}
	
\section{Introduction}
In a recent paper\cite{salvati}, Sylvain Salvati proved that the language 
 \[O_2 = \{ w \in \{a,b,A,B\}^* \mid  |w|_a = |w|_A \wedge |w|_b = |w|_B \}\]
and hence the rationally equivalent language 
 \[\textup{MIX} = \{ w \in \{a,b,c\}^* \mid |w|_a = |w|_b = |w|_c\}\] 
are $2-$Multiple Context Free Languages as follows. Consider the grammar \[G=(\Omega,\mathcal{A},R,S)\] where $\Omega=\left(\{S;\inv\},\rho\right)$ with $\rho(S)=1$ and $\rho(\inv)=2$; $\mathcal{A}=\{a,A,b,B\}$ and $R$ consists of the following rules:
	\begin{enumerate}
		\item $S(x_1x_2) \lar \inv(x_1,x_2)$
		\item $\inv(t_1,t_2)\lar  \inv(x_1,x_2)$ where $t_1t_2\in \perm(x_1x_2aA)\cup \perm(x_1x_2bB)$
		\item $\inv(t_1,t_2)\lar  \inv(x_1,x_2),\inv(y_1,y_2)$ where $t_1t_2\in \perm(x_1x_2y_1y_2)$
		\item $\inv(\epsilon,\epsilon)$
	\end{enumerate} 

Note that by construction it is easy to prove that the language $L=\{w\mid S(w) \text{ is derivable in }G\}$ generated by $G$, is a subset of $O_2$. Using ideas specific to two-dimensional geometry, e.g. complex exponential function, Salvati then proves that 

\begin{thm_2D_derivable}[Salvati]
        If $w_1w_2\in O_2$, then $\inv(w_1,w_2)$ is derivable and hence $w_1w_2$ will be in the language $L=\{w|S(w) \text{ is derivable in }G\}$ generated by $G$.
\end{thm_2D_derivable}

Theorem \ref{thm:2D_derivable} along with the previous description of the grammar $G$ thus proves that $L=O_2$. In this paper, we give a different proof of the same theorem using ideas from homology theory that are not specific to two dimensions. This raises the possibility that the proof might be generalized to higher dimensional cases, thus shedding light on the (still open) question whether $O_n$ is a $n-$MCFL.
 
\section{Acknowledgement}
I would like to thank Danny Calegari, my advisor, for sharing with me Salvati's paper \cite{salvati} which directly inspired the main result of this paper; and for his continued support and guidance, as well as for the extensive comments and corrections on this paper. I would also like to thank Mark-Jan Nederhof for sharing his paper \cite{MJN} and other related works to this problem and Greg Kobele for some useful discussion and comments. 

\section{A generalization of Salvati's Theorem in $2-$dimension}

A word in $O_2$ corresponds uniquely to a closed path in the integer lattice $\Z^2$ as follows. $a$ and $b$ correspond to $\rightarrow$ and $\uparrow$, and $A,B$ to $\leftarrow$ and $\downarrow$ respectively. e.g. the word $abbAbaBaBBBAbA\in O_2$ corresponds to the path in figure \ref{fig:latticepath}.
\begin{figure}[!h]
\centering
	\begin{tikzpicture}	
	
	\begin{scope}[
	decoration={markings, mark=at position 0.9 with {\arrow{>}}}
	]
	\filldraw (0,0) circle (3pt) node[above left]{$(0,0)$};
	\draw[postaction={decorate}] (0,0) -- (1,0);
	\draw[postaction={decorate}] (1,0) -- (1,1);
	\draw[postaction={decorate}] (1,1) -- (1,2);
	\draw[postaction={decorate}] (1,2) -- (0,2);	\draw[postaction={decorate}] (0,2) -- (0,3);
	\draw[postaction={decorate}] (0,3) -- (1,3);
	\draw[postaction={decorate}] (1,3) -- (1,2);
	\draw[postaction={decorate}] (1,2) -- (2,2);
	\draw[postaction={decorate}] (2,2) -- (2,1);
	\draw[postaction={decorate}] (2,1) -- (2,0);
	\draw[postaction={decorate}] (2,0) -- (2,-1);
	\draw[postaction={decorate}] (2,-1) -- (1,-1);
	\draw[postaction={decorate}] (1,-1) -- (1,0);
	\draw[postaction={decorate}] (1,0) -- (0,0); 
	
	\end{scope}
	
	\end{tikzpicture}
	\caption{Lattice path corresponding to the word $abbAbaBaBBBAbA$}
	\label{fig:latticepath}
\end{figure}

We prove a purely topological result about closed curves on the plane (not necessarily simple) which, as a corollary, gives us theorem $\ref{thm:2D_derivable}$.

Consider a closed (oriented) loop $K$ in $\R^2$ given by 
\[\phi: S^1\to \R^2,\quad Image(\phi)=:K\]

It makes no difference to set the domain of $\phi$ equal to $[0,1]$ and assume $\phi(0)=\phi(1)$. Let $p=\phi(0), q=\phi(1/2)$ and let $r$ and $s$ be two arbitrary points on $K$. Together the four points $\{p,q,r,s\}$ break up $K$ into $4$ (possibly degenerate) arcs $K_1,K_2,K_3$ and $K_4$ such that the starting point of $K_{i+1}$ is the same as the ending point of $K_i$.  Denote by $\v_i$ the vector which is defined as 
\[\v_i=\text{End point of }K_i -  \text{ Starting point of }K_{i}\]
The vectors $\v_i$ satisfy 
\[\vec{v}_1+\vec{v}_2+\vec{v}_3+\vec{v}_4=0\] 
Our main technical result is the following.
\begin{proposition}\label{prop:loop}
 Assume that $\phi$ is differentiable at $p$ and $q$ and $\phi'(0)$ is not antiparallel to $\phi'(1/2)$. Then there exists a pair of $r$ and $s$ on $K$ different from $\{p,q\}$ such that  the set of $4$ vectors $\{\vec{v}_1,\vec{v}_2,\vec{v}_3,\vec{v}_4\}$ can be partitioned into two sets of $2$ vectors each of which sum up to zero.
\end{proposition}
 
The proof of Proposition \ref{prop:loop} proceeds in multiple steps as follows. First we construct a $2$ dimensional cell complex $\mc{X}$ that parametrizes the choices for $r$ and $s$ along with a choice of partition for $\{\vec{v}_1,\vec{v}_2,\vec{v}_3,\vec{v}_4\}$. Then we define a function $f:\mc{X}\to\R^2$ with the property that a zero for $f$ gives a choice of $\{r,s\}$ satisfying the conclusion of proposition \ref{prop:loop}. Finally using a homological argument, we show that $f$ must have a zero.

Before moving on with the proof we make the following observations to simplify the proof. We claim that it suffices consider embedded $K$. Suppose the loop $K$ has at least one self intersection i.e. it's not embedded. Then there exist two points $0\leq t_1<t_2\leq 1$ such that $\phi(t_1)=\phi(t_2)$. Now depending on values of $t_i$, we must have one of the following cases (recall that $p$ and $q$ are distinct):
\begin{enumerate}
    \item $ t_1=0$   and $0<t_2<1/2$
    
    Then we remove $\phi(0,t_2)$ of $K$ to get a new knot $K'$. Clearly if $K'$ satisfies the conclusion of proposition \ref{prop:loop} then so does $K$. So we replace $K$ with this new simplified knot $K'$.
    
    \item $ t_1=0$   and $1/2<t_2<1$
    
    In this case we remove $\phi(t_2,1)$ instead.
    
    \item $ t_2=1/2$   and $0<t_1<1/2$ OR $ t_1=1/2$   and $1/2<t_2<1$ 
    
    Similar to the above two cases.
    
    \item $ 0<t_1<t_2<1/2$ OR $ 1/2<t_1<t_2<1$
    
    Remove $\phi(t_1,t_2)$.
    
    \item $ 0<t_1<1/2<t_2<1$
    
    Then choose $r=\phi(t_1)$ and $s=\phi(t_2)$ so that $\vec{v}_1+\vec{v}_4=\vec{v}_2+\vec{v}_3=0$.
    
\end{enumerate} 
In conclusion, after necessary simplifications, without loss of generality we can assume that $K$ is an embedded loop in $\R^2$.

\subsection{Construction of the $2$ dimensional Cell Complex  $\mc{X}$ and the function $f$}\label{subsec:construction}

In this section, we define the $2$-complex $\mc{X}$ whose points parametrizes choices of $\{r,s\}\neq\{p,q\}$ on $K$ along with the particular choice of partition of $\{\vec{v}_1,\vec{v}_2,\vec{v}_3,\vec{v}_4\}$ into two sets of $2$ vectors each. Note that the choice of partition also defines the function $f$ on $\mc{X}$, but this is a bit complicated since it depends on the relative orders of $\{p,q,r,s\}$ on $K$. Now choosing $r$ and $s$ on $K$ is equivalent to choosing two numbers $x$ and $y$ in the interval $[0,1]$, where $0$ and $1$ are considered the same number. In what follows, we list all the possible $2$-cells in $\mc{X}$ in table \ref{2cells}; along with the $1$-cells and $0$-cells that appear as the boundary of the $2$-cells and parametrize the cases when our choice is degenerate.  This will take up next couple of pages. However, note that for our subsequent sections and to prove the main result, we will be using only some of these cells.

The $2$-cells fall into three categories:
\begin{enumerate}[\textbf{Case} 1.]
	\item $0\leq x\leq 1/2\leq y\leq 1$
	\item $0\leq x\leq y\leq 1/2$
	\item $1/2\leq x\leq y\leq 1$
\end{enumerate}
In each case, we define $f$ in such a way that $f(x,y)=0$ implies existence of a partition of $\{\vec{v}_1,\vec{v}_2,\vec{v}_3,\vec{v}_4\}$ into two sets of two vectors each of which add up to zero. Observe that if $(p,r,s,q)$ appear in that particular order in $K$ and $\vec{v}_2=s-r=0$ (which is equivalent to saying $\vec{v}_1+\vec{v}_3+\vec{v}_4=0$), then that means part of $K$ from $\phi(x)$ to $\phi(y)$ makes a loop. But since we assumed $K$ is embedded, those combination of vector(s) are never zero for any choice of $x$ and $y$ in case $2$. So we can safely add the cells $F$ and $H$ in $\mc{X}$. Similarly, when $p,q,r,s$ appear in that particular order in $K$ we can consider the extra $2$-cells $J$ and $L$.

\newcolumntype{C}[1]{>{\centering\arraybackslash}m{#1}}
\newcolumntype{N}{@{}m{0pt}@{}}
\begin{longtable}{|C{0.05\textwidth}|C{0.23\textwidth}|C{0.17\textwidth}||C{0.05\textwidth}|C{0.23\textwidth}|C{0.17\textwidth}|N}
	\caption{The list of $2$-cells}\label{2cells}\\
	\hline
	\!\Small{Name} & picture of the cell & $f(x,y)$ on the cell & \!\Small{Name} & picture of the cell & $f(x,y)$ on the cell &\\[10pt]
	\hline 
	\hline 
	\endfirsthead
	\caption{(continued)}\\
	\hline
	\!\Small{Name} & picture of the cell & $f(x,y)$ on the cell & \!\Small{Name} & picture of the cell & $f(x,y)$ on the cell &\\[10pt]
	\hline
	\endhead
	\multicolumn{7}{|c|}{$2$-\textsc{cells}}\\
\hline
\end{longtable}
Before writing down the boundary maps, we make a final observation. In the cells $E, F,G,$ and $H$, the domain of $f$ is as in figure \ref{fig:soln_spc_2} and the function $f$ is constant on the hypotenuse $\{x=y\}$ of the underlying triangle. 
\begin{figure}[!ht]
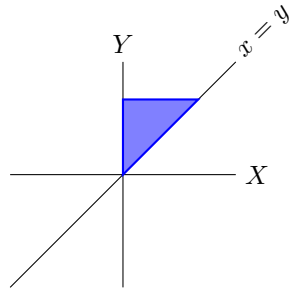

	\centering
	%
	\caption{Solution space for case $2$}
	\label{fig:soln_spc_2}
\end{figure}
Hence for the purpose of writing down the boundary map, we can replace the cells with homeomorphic bigons whose boundary is given by $1-$cells corresponding to the sides $\{x=0\}$ and  $\{y=1/2\}$ of the initial triangles (and the third side is contracted to a point) [cf. figure \ref{fig:soln_space_2b}].
            \begin{figure}[!ht]
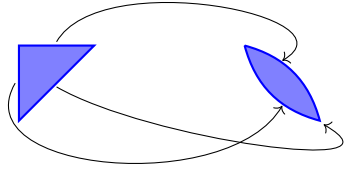

                \centering
                %
                \caption{Homeomorphic solution space for case $2$}
                \label{fig:soln_space_2b}
            \end{figure}

By abuse of notation, we will denote these homeomorphic $2$-cells by the same name. By symmetry, same thing is true for the four cells $I,J,K,$ and $L$. Then the boundary maps on the $2$-cells  are defined as follows.
        \begin{align*}
            \d(A) &= \widebar{\beta} - \alpha - \widebar{\delta} + \gamma &\d(B) &= \widebar{\alpha} - \beta - \gamma + \widebar{\delta} & & &\\
            \d(C) &= \beta - \widebar{\alpha} - \delta + \widebar{\gamma} & \d(D) &= \alpha - \widebar{\beta} - \widebar{\gamma} + \delta\\
            \d(E) &= \alpha + \beta & \d(F) &= \beta + \alpha &
            \d(G) &= \widebar{\alpha} + \widebar{\beta} & \d(H) &= \widebar{\beta} + \widebar{\alpha}\\
            \d(I) &= \gamma + \delta & \d(J) &= \delta + \gamma &
            \d(K) &= \widebar{\gamma} + \widebar{\delta} & \d(L) &= \widebar{\delta} + \widebar{\gamma}
        \end{align*}

And finally the boundary maps from $1-$cells to $0-$cells are given by

\begin{align*}
            \d(\alpha) &= p_2 - p_1 & \d(\beta) &= p_1 - p_2 &
            \d(\gamma) &= p_3 - p_1 & \d(\delta) &= p_1 - p_3\\
            \d(\widebar{\alpha}) &= p_3 - p_4 & \d(\widebar{\beta}) &= p_4 - p_3 &
            \d(\widebar{\gamma}) &= p_2 - p_4 & \d(\widebar{\delta}) &= p_4 - p_2\\
\end{align*}
        
\begin{figure}[!ht]
                \centering
                \def\svgwidth{\textwidth}
                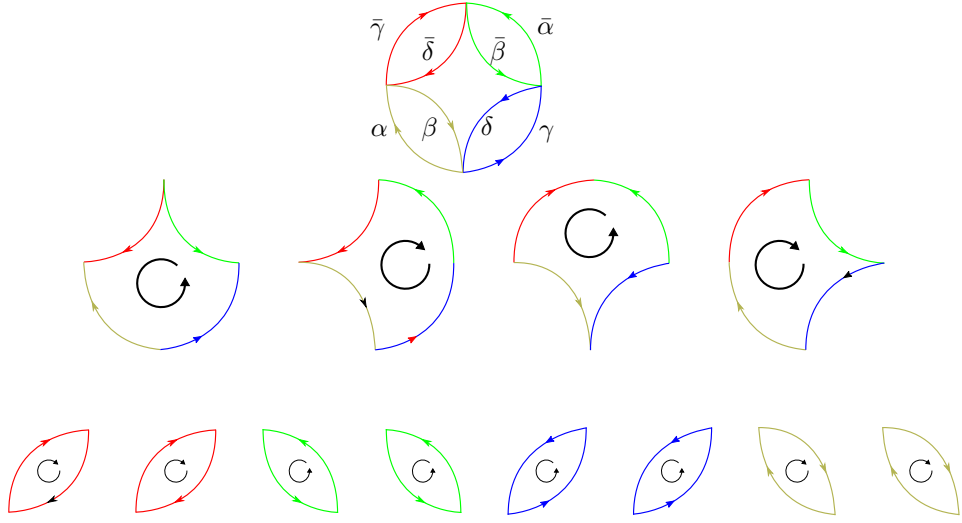
                \caption{All cells with orientation}
                \label{fig:X2cells}
\end{figure}

Refer to figure $\ref{fig:X2cells}$ to see how all the cells glue together with correct orientation.

\subsection{Taking out some of the vertices of $\mc{X}$ to get a homologically trivial $2$-cycle}

So far we have constructed a $2-$dimensional complex, let's call it $\mc{X}$, with four $0-$cells, eight $1-$cells and twelve $2-$cells. We can picture $\mc{X}$ as in figure $\ref{fig:X2cells}$.

In last section, we defined a continuous vector-valued function from $\mc{X}$ to $\R^2$. Note that $f(p_1)$ and $f(p_4)$ are zero, however they correspond to the degenerate case when $\{p,q\}=\{r,s\}$. So to prove proposition \ref{prop:loop}, we need to show that $f$ has a nontrivial zero on $\mc{X}$ other than $p_1$ and $p_4$. 

Let us assume, for the sake of contradiction, that $f$ is nonzero everywhere else on $\mc{X}$. We remove two circular $\epsilon$-neighborhoods of $p_1$ and $p_4$ from $\mc{X}$ and compactify (in the obvious way) the resulting space to get a new cell complex $\mc{Y}$. Clearly we have
\[\d\mc{Y}= \Lk(p_1,\mc{X}) \cup \Lk(p_4,\mc{Y})\]
where $\Lk(p_1,\mc{X})$ is the link of $p_1$ in $X$. Now $f|_{\mc{Y}}$ is nonzero and so we can replace $f$ by $f/\|f\|$, which by abuse of notation, we will still call $f$, to get a continuous function defined as \[f:\mc{Y} \to S^1.\]
Then we have a commutative diagram
\[\begin{tikzcd}
\d\mc{Y} \arrow[r,hookrightarrow, "i"] \arrow[d , "i\circ f"'] & \mc{Y} \arrow[ld, "f"]\\
S^1 & \\
\end{tikzcd}\] which induces following commutative diagram on the homology of these spaces
\[\begin{tikzcd}
H_1(\d\mc{Y}) \arrow[r] \arrow[d] & H_1(\mc{Y}) \arrow[ld]\\
H_1(S^1)=\Z & \\
\end{tikzcd}\]

Our goal is to show existence of cycles in $H_1(\d\mc{Y})$ that are homologcally trivial in $H_1(\mc{Y})$ but map nontrivially to $H_1(S^1)$. [For a short introduction to the theory of homology, the reader may refer to chapter $2$ of \cite{Hatcher}.]

Recall that $\Lk(p_1,\mc{X})$ (and similarly $\Lk(p_4,\mc{X})$) is a graph whose vertices correspond to the $1-$cells of $\mc{X}$ which are incident to $p_1$ (resp. $p_4$). Two such edges are adjacent in $\Lk(p_1,\mc{X})$ (resp. $\Lk(p_4,\mc{X})$) if they are incident to a common $2-$cell at $p_1$ (resp. $p_4$). From the list of cells above we can then easily see that both the links have has $4$ vertices and $8$ edges (and the two graphs are disjoint from each other).

Now consider the cycle in $\Lk(p_1,\mc{X})$ that consists of the $1-$cell (edge in the graph) corresponding to $E$ minus the $1-$cell corresponding to $F$. Since the $2-$cells clearly bound a sphere in $\mc{X}$, this cycle as an element of $H_1(\mc{Y})$ bounds a disk, and hence is homologically trivial. We will call this the $\alpha\beta-$cycle because of the boundary of the $2-$cells.  Similarly we can construct $\gamma\delta-, \widebar{\alpha}\widebar{\beta}-$ and $\widebar{\gamma}\widebar{\delta}-$cycle in $H_1(\d\mc{Y})$ which map homologically trivially into $H_1(\mc{Y})$. In general we can produce such a cycle whenever a collection of $2-$cells of $\mc{X}$ makes a sphere.

\subsection{Calculating degree of $f$ on cycles of $\d\mc{Y}$}\label{subsec:degreecalc}

We first concentrate on the $\epsilon-$neighbourhood, let's call it  $U$, of the point $p_1$ in $\mc{X}$ that was used to obtained $\Lk(p_1,\mc{X})$ as $\d U$. Similarly, we deonte the corresponding neighbourhood at $p_4$ by $V$. We want to choose $U$ and $V$ sufficiently small such that the following is true.

Suppose $U\cap \alpha = [0,a]$ and $U\cap \delta = [d,1]$. We want the vectors $\vec{u}_\alpha= \phi(a) - \phi(0)$ and $\vec{u}_\delta= \phi(1)-\phi(d)$ approximately of same magnitude and in the same direction (note that $\phi(0)=\phi(1)$). Similarly  if $U \cap \beta=[x,1/2]$ and $U \cap \gamma=[1/2,c]$, then we want $\vec{u}_\beta=\phi(1/2)-\phi(b)$ and $\vec{u}_\gamma=\phi(c)-\phi(1/2)$ to be of approximately same magnitude and direction. Now $f$ is differentiable in a neighbourhood of $p_1$ since $\phi$ is differentiable at $p$ and $q$. So such a choice of $U$ exists. Note that we can simultaneously choose $V$ small enough such that the following are also satisfied. 
 \[\vec{u}_{\widebar{\alpha}}\approx \vec{u}_{\widebar{\delta}},\, \vec{u}_{\widebar{\beta}} \approx \vec{u}_{\widebar{\gamma}}\] and \[\vec{u}_{\widebar{\alpha}}=-\vec{u}_{\alpha},\vec{u}_{\widebar{\beta}}=-\vec{u}_{\beta}, \vec{u}_{\widebar{\gamma}}=-\vec{u}_{\gamma} \text{ and } \vec{u}_{\widebar{\delta}}=-\vec{u}_{\delta}.\]

Let $\theta$ denote the angle from the vectors $\vec{u}_\alpha$ to $\vec{u}_\beta$ and $\theta'$  denote the angle from the vectors $\vec{u}_\gamma$ to $\vec{u}_\delta$ (in clockwise direction, wlog). Since $\vec{u}_\alpha \approx \vec{u}_\delta $ and $\vec{u}_\beta \approx \vec{u}_\gamma$, one of the two angles $\theta$ and $\theta'$ is less than or equal to $\pi$. Without loss of generality let us assume that $\theta<\pi$. Then we know that the angle from $\vec{u}_{\widebar{\delta}}$ to $\vec{u}_{\widebar{\gamma}}$ is also less than $\pi$.

Observe that the vectors $u_\alpha$ and $u_\beta$ are essentially the direction of tangents to the knot $K$ at $0$ and $1/2$. Consequently since $K$ is embedded, up to homotopy, the part of $\phi$ from $0$ to $1/2$ looks like one of the following $3$ possibilities in figure \ref{fig:3homotopycases}. 

\begin{figure}[!ht]
	\centering
	\begin{subfigure}[t]{0.33\textwidth}
		 \centering
		 \includegraphics[scale=0.33]{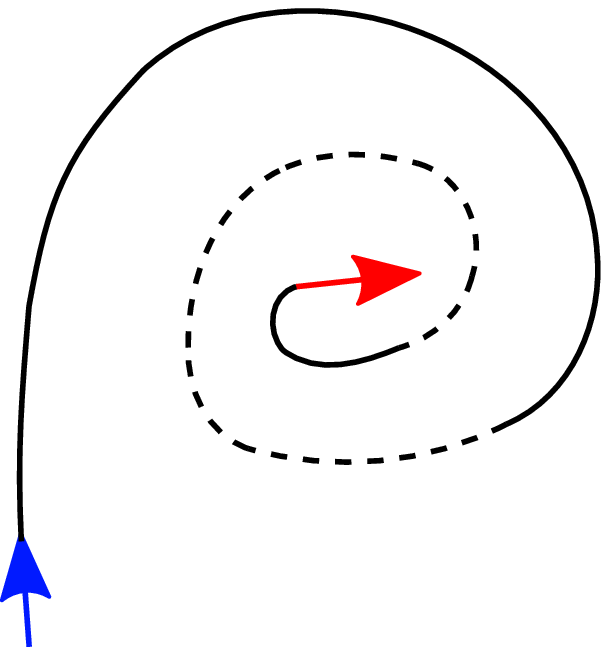}
		 \caption{Case $1$}
	\end{subfigure}%
	\begin{subfigure}[t]{0.33\textwidth}
		 \centering
		 \includegraphics[scale=0.33]{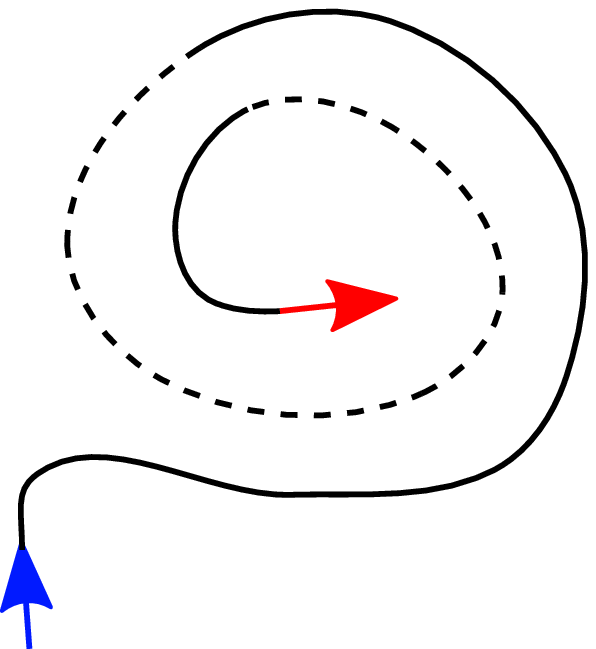}
		 \caption{Case $2$}
	\end{subfigure}%
	\begin{subfigure}[t]{0.33\textwidth}
		\centering
		\includegraphics[scale=0.33]{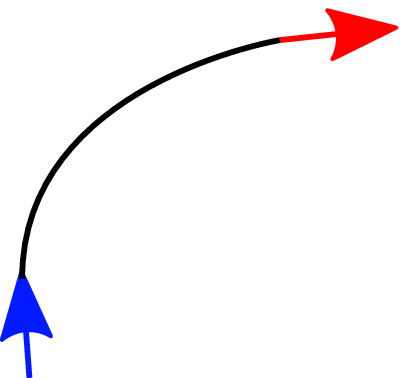}
		\caption{Case $3$}
	\end{subfigure}
	\caption{$3$ Homotopy cases for part of $K$ from $0$ to $1/2$} 
	\label{fig:3homotopycases}
\end{figure}

Let  $p_\alpha, p_\beta, p_\gamma$ and $p_\delta$ denote the vertices in $\Lk(p_1,\mc{X})$ corresponding to the $1-$cells $\alpha,\beta,\gamma$ and  $\delta$ respectively (cf. figure \ref{fig:fancynbhd}).
       
\begin{figure}[!ht]
                \centering
                \def\svgwidth{0.5\textwidth}
                \vspace{20pt}
                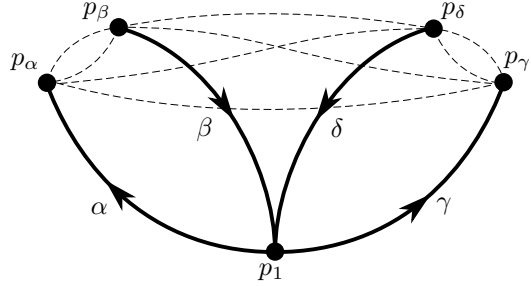
                \vspace{10pt}
                \caption{The link (dotted part) of the vertex $p_1$ in $\mc{X}$}
                \label{fig:fancynbhd}
\end{figure}
We introduce the notation $[p_\alpha\to p_\beta]$ to denote the edge in from $p_\alpha$ to $p_\beta$ in $\Lk(p_1,\mc{X})$ corresponding to the $2-$cell $F$ (the orientation is derived from $\mc{X}$ and the fact that $E-F$ is a  sphere). Similarly we define $[p_\beta\to p_\alpha]$ to denote the edge corresponding to $E$. Thus the $\alpha\beta-$cycle is equal to $[p_\alpha\to p_\beta]+[p_\beta \to p_\alpha]$.

We wish to calculate the degree of of $f|_{\alpha\beta-cycle}$. Depending on the $3$ cases for the embedding of $\phi|_{[0,1/2]}$ as discussed above, we will get different results.

Let's look at $f|_{[p_\beta\to p_\alpha]}$ first. Note that for small enough $U$ i.e for $x\to 0+$ and $y\to 1/2-$, the function $f$ on $[p_\beta \to p_\alpha]$ is approximately a linear interpolation from $f(p_\beta)=\vec{u}_{\beta}$ to $f(p_\alpha)=\vec{u}_{\alpha}$ and hence the image of $f$ traverses the circle from the direction of $\vec{u}_\beta$ to that of $\vec{u}_\alpha$ (cf. figure \ref{fig:p_betato_p_alpha}).

\begin{figure}[!ht]
	\centering
	\includegraphics[scale=0.45]{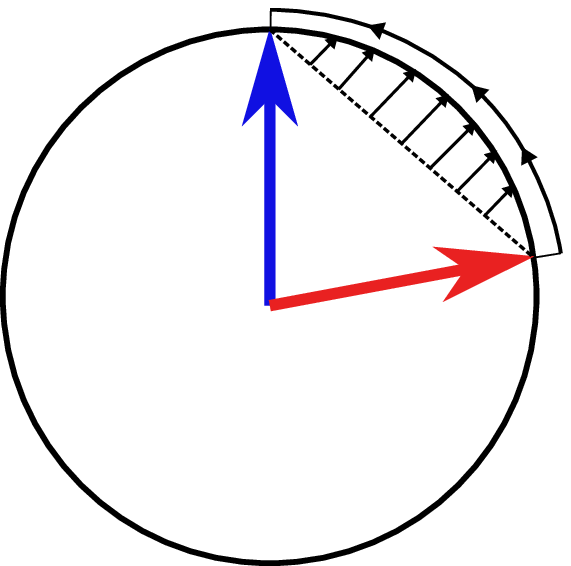}
	\caption{Image of $f$ on $[p_\beta\to p_\alpha]$}
	\label{fig:p_betato_p_alpha}
\end{figure}

The case of $[p_\alpha\to p_\beta]$ is very different because of the the corresponding $2-$cell $F$. Note that by construction, a small neighbourhood of $p_1$ in $F$ looks like \[\{(x,y)\mid y-x\leq \epsilon \} \subset \{0\leq x<y\leq 1/2\}\] for some small $\epsilon>0$ so that $\phi(y)-\phi(x) \approx 0$. So, as we traverse $[p_\alpha \to p_\beta]$, the function $f$ maps continuously from $\phi(\epsilon)-\phi(0)$ to $\phi(1/2)-\phi(1/2-\epsilon)$ through vectors of the form $\phi(x+\epsilon)-\phi(x)$ with $x$ going from $0$ to $1/2-\epsilon$. Thus the winding number of image of $f$ on $[p_\alpha\to p_\beta]$ depends on the embedding of $\phi$ from $0$ to $1/2$. Looking at the $3$ possible cases as in figure \ref{fig:3homotopycases}, we then observe that the image of $f$ looks as in figure \ref{fig:windingcases_AB_123}. Consequently the degree of $f$ on the $\alpha\beta-$cycle can be tabulated as in table \ref{table:degreeAB}.

\begin{figure}[!ht]
	\centering
	\begin{subfigure}[t]{0.33\textwidth}
		\centering
		\includegraphics[scale=0.45]{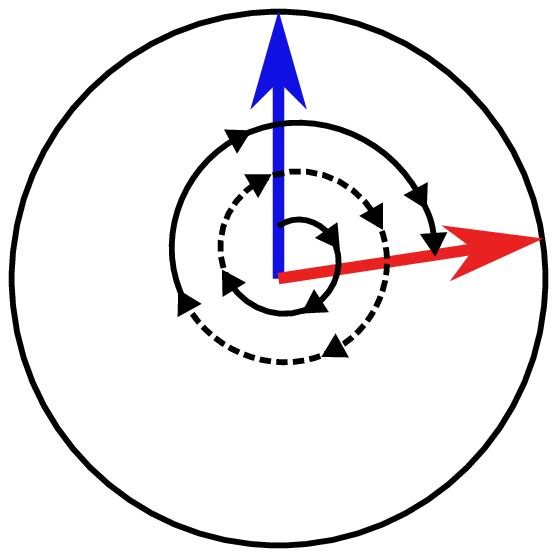}
		\caption{Case $1$}
	\end{subfigure}%
	\begin{subfigure}[t]{0.33\textwidth}
		\centering
		\includegraphics[scale=0.45]{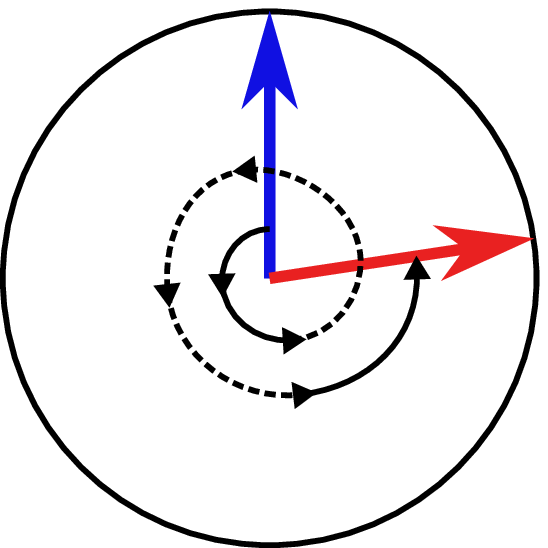}
		\caption{Case $2$}
	\end{subfigure}%
	\begin{subfigure}[t]{0.33\textwidth}
		\centering
		\includegraphics[scale=0.45]{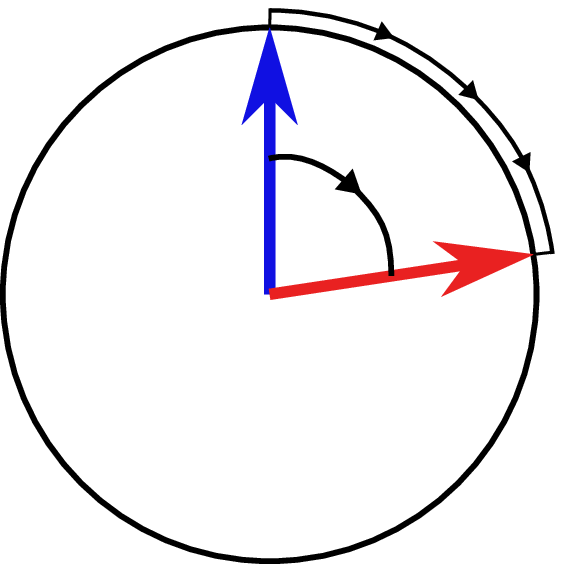}
		\caption{Case $3$}
	\end{subfigure}
	\caption{Image of $f$ in the $3$ cases} 
	\label{fig:windingcases_AB_123}
\end{figure}
\begin{table}[!htb]
	\centering
	\begin{tabular}{|c|c|c|c|}
		\hline \rule[-2ex]{0pt}{5.5ex}  & Case $1$ & Case $2$ & Case $3$ \\ 
		\hline \rule[-2ex]{0pt}{5.5ex} Degree of $f$ on $\alpha\beta-$cycle & $\gneq 0$ & $\lneq 0$ & $=0$ \\ 
		\hline 
	\end{tabular} 
	\caption{}\label{table:degreeAB}
\end{table}

In the case when the degree of $f$ on the $\alpha\beta-$cycle is zero, we look at the possibilities for the part of the knot $\phi$ from $1/2$ to $1$. Again up to homotopy, there are exactly $2$ possible cases as in figure \ref{fig:knotorientations}. In each case, let's calculate the degree of $f$ on the $\gamma\delta-$cycle which depends on the embedding of $\phi$ from $1/2$ to $1$. Then it is easy to see that in both cases, $f|_{\gamma\delta-cycle}$ is non zero.

\begin{figure}[!ht]
	\centering
	\begin{subfigure}[t]{0.5\textwidth}
		\centering
		\includegraphics[scale=0.33]{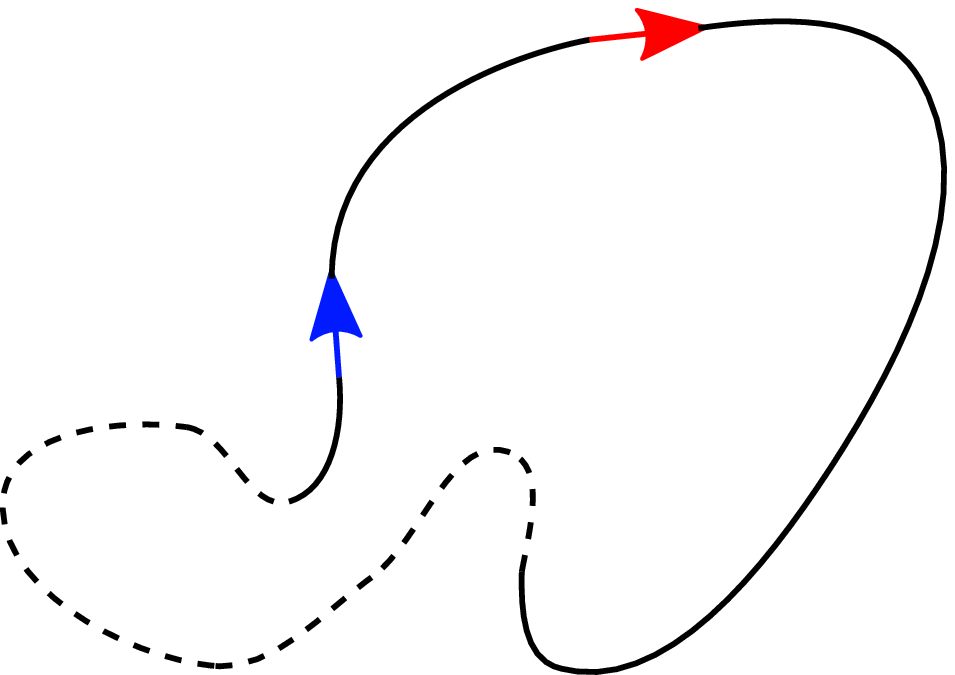}
		\caption{Case $1$}
	\end{subfigure}%
	\begin{subfigure}[t]{0.5\textwidth}
		\centering
		\includegraphics[scale=0.33]{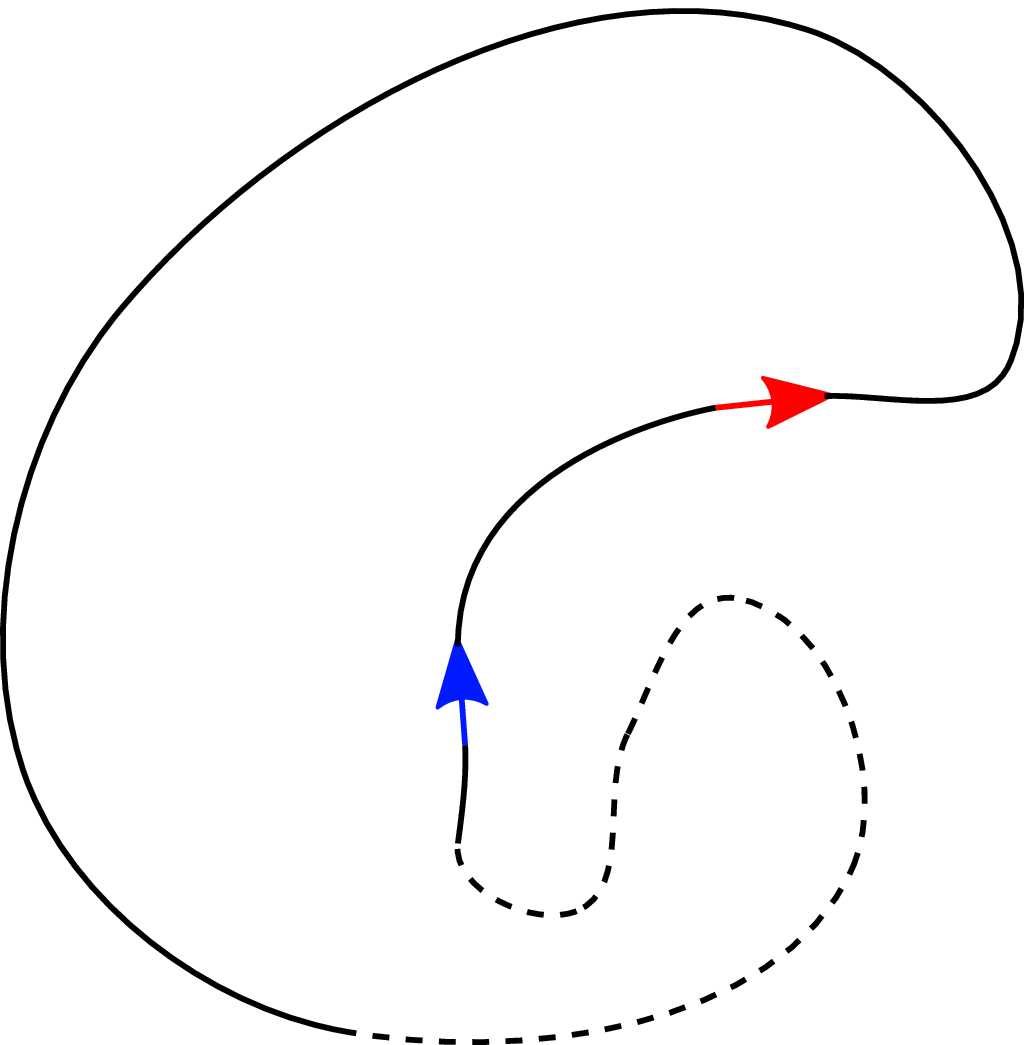}
		\caption{Case $2$}
	\end{subfigure}
	\caption{$2$ Homotopy cases for part of $K$ from $1/2$ to $1$ for case $3$ in figure \ref{fig:3homotopycases}} 
	\label{fig:knotorientations}
\end{figure}

\subsection{Conclusion}

The proof of proposition $\ref{prop:loop}$ now follows from the following contradiction. Since the $\alpha\beta-$cycle and the $\gamma\delta-$cycle are homologically trivial in $\mc{Y}$, the composition \[H_1(\d \mc{Y}) \hookrightarrow H_1(W) \to H_1(S^1)\] is  a zero map. But on the other hand at least one of these cycle maps with nonzero degree to $S^1$ and hence maps nontrivially in homology. Thus our assumption in subsection $\ref{subsec:degreecalc}$ is wrong and $f$ indeed must have a nontrivial zero other than $p_1$ and $p_4$. \hfill $\qed$

\section{Proof of Salvati's Theorem}

Following Salvati, we first prove the following corollary to proposition $\ref{prop:loop}$ which after a easy induction argument will prove theorem $\ref{thm:2D_derivable}$.

\begin{theorem} \label{thm:induction_argument}
        Let $w_1,w_2\in \{a,b,A,B\}^*$ such that $w_1w_2\in O_2$. Then there exists $x_1,x_2,y_1,y_2\in\{a,b,A,B\}^*$ such that 
        \begin{enumerate}
            \item $w_1w_2\in perm(x_1x_2y_1y_2)$, 
            \item $\max\{|x_1x_2|,|y_1y_2|\} < |w_1w_2|$,
            \item $w_1w_2\in perm\{x_1x_2y_1y_2\}$, and
            \item $x_1x_2$ and $y_1y_2$ are both in $O_2$.
        \end{enumerate}
\end{theorem}
     
\begin{proof}
    We consider the set of all possible choices of $x_1,x_2,y_1$, and $y_2$ which satisfy condition $(1),(2),$ and $(3)$ from above. Let's call this set $S$. It is easy to see that $S$ is nonempty and finite. We want to show that at least one of these choices will also satisfy condition $(4)$. 
     
    Recall that each word in $O_2$ and in general any word in $a,b,A,B$ corresponds to a path in the integer lattice in $\R^2$. By abuse of notation, we will denote the path corresponding to $w$ by $w$ as well. Then let $\v(w)$ be the vector in $\R^2$ which goes from the starting point of $w$ to the end point of $w$. We will use the notation $(a,b)$ to mean the point or the vector $\la a,b\ra=a\hat{i}+b\hat{j}$ interchangeably and the meaning should be clear from the context. Thus,
     \[\v(w)=|w|_a\la 1,0\ra + |w|_b\la0,1\ra + |w|_A\la-1,0\ra + |w|_B\la0,-1\ra\]
   
    We will say a point $(a,b)$ is an lattice point if both $a,b\in \Z$. Similarly we say a vector $\vec{v}$ is an lattice vector if both its $\hat{i}$ component and $\hat{j}$ components are integer. Thus the vector $(a,b)$ is an lattice vector if it is an lattice point.

    For each choice in $S$, we observe that for $(x_1,x_2,y_1,y_2)$ in $S$ to satisfy $(4)$ is equivalent to having $\v(x_1)+\v(x_2)$ equal to the zero vector i.e. the origin. Note that by construction, the image under $P$ always gives an lattice vector.
    
    We would like to apply proposition $\ref{prop:loop}$ to the closed lattice path corresponding to $w=w_1w_2$, where the two fixed points $p$ and $q$ are the starting point of $w_1$ (i.e. the ending point of $w_2$ and also, the origin) and the starting point of $w_2$ (which is also the ending point of $w_1$). Unfortunately, the closed lattice path is clearly not differentiable at those two points. We circumvent this problem in the following way.
    
    Firstly note that if either the starting or ending letter in $w_1$ is inverse of either the starting or ending letter in $w_2$, then we can easily find $x_1,x_2x_3,x_4$ as above For example if $w_1=aw_1'b$ and $w_2=bw_2'A$, then $x_1=a, x_2=A, x_3=w_1'b$, and $ x_4=bw_2'$ will work. So we might as well assume that $w_1$ and $w_2$ do not have compatible letters at the starting or ending. This takes care of the condition that the tangents at $p$ and $q$ are not antiprallel.
    
    Let's write $w_1=s_1w_1's_2$ and $w_2=s_3w_2's_4$, where $s_i\in\{a,b,A,B\}$ is a letter. Here we are assuming $|w_1|,|w_2|\geq 2$ since otherwise the proof is trivial. Now let's say $s_1=a$. Then none of $s_2,s_3,$ and $s_4$ can be $A$. Also $b$ and $B$ both can not be in the set $\{s_2,s_3,s_4\}$. So we also assume $\{s_2,s_3,s_4\}=\{a,b\}$. The other cases are similar. Now the two possibilities for $s_4s_1$, i.e. the two letters around the point $p$ are $aa$ or  $ab$. For each of those possibilities $s_2s_3$, the two letter aroud $q$ can be one of the four \[ab,aa,bb,ba\]
    
    \begin{figure}[!ht]
            \centering
            \includegraphics[width=0.75\textwidth]{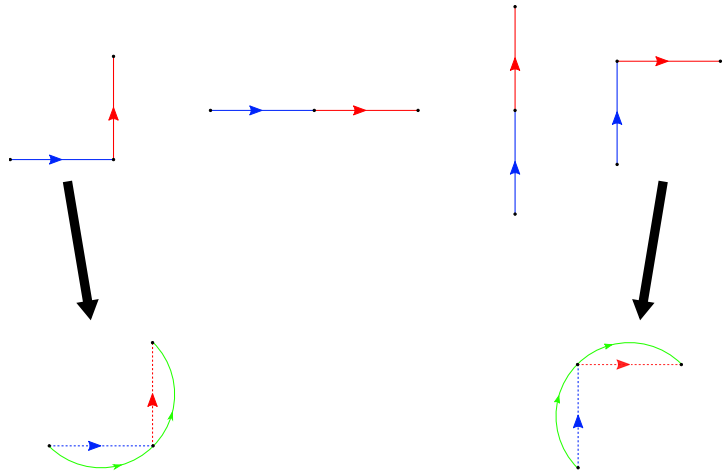}
            \caption{`Smooth'ening the path at $p$ and $q$}
            \label{fig:smoothing}
    \end{figure}

    For any of the above $8$ possibilities, we replace the lattice path corresponding to $ab$ and $ba$ (only at $p$ or $q$)  by a smooth circular arc as in figure \ref{fig:smoothing}, so that smooth path still passes through the lattice points but none of the other points in the arc is an lattice point. If either of $s_4s_1$ and $s_2s_3$ is $aa$ or $bb$ then we leave it as it is. Thus we have changed our closed lattice path corresponding to $w$ into another loop which is differentiable at $p$ and $q$, and so we can now apply prop $\ref{prop:loop}$.
    
    The proposition then gives us the existence of two points $r$ and $s$ which in turns gives us four vectors $\vec{v}_1,\vec{v}_2,\vec{v}_3,$ and $\vec{v}_4$, two of which add up to $0$. It is easy to see that if $r$ and $s$ is actually an lattice point, then all four of these vectors are lattice vectors and we can easily use $\vec{v}_1,\vec{v}_2,\vec{v}_3,\vec{v}_4$ suitably to get a choice of $(x_1,x_2,y_1,y_2)$ in $S$ which satisfies $(4)$. 
    
    \begin{figure}[!ht]
            \centering
            \includegraphics[width=0.75\textwidth]{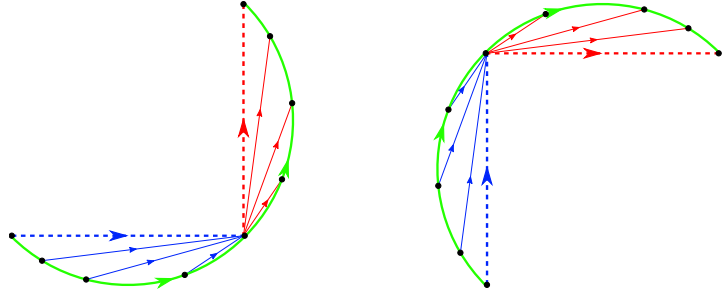}
            \caption{$r$ or $s$ cannot lie in the replaced part}
            \label{fig:samequadrant}
    \end{figure}
    
    Another interesting observation is that $r$ and $s$ cannot lie in the part of the lattice path we replaced in the previous paragraph to `smooth'en the loop at $p$ and $q$. Because otherwise the vectors $\vec{v}_1,\vec{v}_2,\vec{v}_3,\vec{v}_4$ would all lie in the same quadrant and cannot sum up to $0$. (cf figure \ref{fig:samequadrant}). So all we need to prove is that we can find points $r$ and $s$ satisfying proposition $\ref{prop:loop}$ which do no lie in the middle of a horizontal (corresponding to $a$ or $A$) or a vertical (corresponding to $b$ or $B$) segment in the modified closed lattice path.
    
    Contrary to what we have to prove, suppose all choices of $r$ and $s$ satisfying the conclusion of proposition $\ref{prop:loop}$ have at least one of them a non-lattice point. We can assume wlog that $r$ (resp. $s$) lies in the path corresponding to $w_1$(resp. $w_2$). All the other possibilities can be dealt with very similar and easy arguments.
    
    In the first case, for some choice of $r$ and $s$ as above, suppose $r$ in an lattice point and $s$ lies in the interior of a segment.  Since $p$ and $q$ are lattice points by our construction, $\vec{v}_1$ and $\vec{v}_2$ are lattice vectors but $\vec{v}_3$ and $\vec{v}_4$ are not. Then the partition mentioned in proposition $\ref{prop:loop}$ must be $\{\vec{v}_1,\vec{v}_2\}$ and $\{\vec{v}_3,\vec{v}_4\}$ because otherwise they don't add up to $0$(which is an lattice vector). But then we might as well choose $s$ to coincide with $q$ and then both $r$ and $s$ become lattice points and still satisfy prop $\ref{prop:loop}$. Contradiction.
    
    Next suppose both $r$ and $s$ are not lattice points. Again if our choice for partition is $\{\vec{v}_1,\vec{v}_2\}$ and $\{\vec{v}_3,\vec{v}_4\}$ then we can replace $r$ and $s$ with $p$  and $q$ respectively and be done. Suppose $r$ lies in the middle of a horizontal segment corresponding to $a$. Then the only way either $\vec{v}_1+\vec{v}_3=0$ or $\vec{v}_1+\vec{v}_4=0$ is if $s$ also lies in the middle of a horizontal segment, since otherwise $\vec{v}_1+\vec{v}_3$(or $\vec{v}_1+\vec{v}_4$) can not be an lattice point. Moreover, $r$ and $s$ both must divide the segments they belong to in equal or opposite ratios. Then we suitably  replace $r$ and $s$ by the starting or end point of the segments they are in, since doing so adds and subtracts same vector and so the sum remains unchanged (stil equal to zero). This gives us  a choice of  lattice points $r$ and $s$ which also satisfy the conclusion of prop. \ref{prop:loop}. Contradiction.
    
    Thus either way we can always find lattice points $r$ and $s$ and consequently a choice of $(x_1,x_2,y_1,y_2)\in S$ satisfying $(4)$ in theorem $\ref{thm:induction_argument}$.

\end{proof}

\begin{theorem}\label{thm:2D_derivable}
    If $w_1w_2\in O_2$, then $\inv(w_1,w_2)$ is derivable and hence $w_1w_2$ will be in the language $L=\{w|S(w) \text{ is derivable in }G\}$ generated by $G$.
\end{theorem}

\begin{proof}
    We prove the theorem by inducting on the size of  the word $w_1w_2$. The base cases are trivial and follow from the rules in $G$. Theorem \ref{thm:induction_argument} and the third rule in $G$ forms the induction step.
\end{proof}

\section{The case of $3$ dimension and higher}
We would like to generalize above proof to the case of $O_3$  and conjecture the following.
\begin{conjecture}\label{conj:O_3}
$O_3$ is a $3-$MCFL and is generated by the following grammar. Consider  \[G=(\Omega,\mathcal{A},R,S)\] where $\Omega=\left(\{S;\Circ\},\rho\right)$ with $\rho(S)=1$ and $\rho(\Circ)=3$; $\mathcal{A}=\{a,A,b,B,c,C\}$ and $R$ is made of rules that have one of the following forms:
	\begin{enumerate}
		\item $S(x_1x_2x_3) \lar \Circ(x_1,x_2,x_3)$
		\item $\Circ(t_1,t_2,t_3)\lar  \Circ(x_1,x_2,x_3)$
		
		where $t_1t_2t_3 \in \perm(x_1x_2x_3aA)  \cup \perm(x_1x_2x_3bB) \cup \perm(x_1x_2x_3cC)$
		\item $\Circ(t_1,t_2,t_3)\lar  \Circ(x_1,x_2,x_3),\Circ(y_1,y_2,y_3)$
		
		where $t_1t_2t_3\in \perm(x_1x_2x_3y_1y_2y_3)$
		\item $\Circ(\epsilon,\epsilon,\epsilon)$
	\end{enumerate} 

\end{conjecture}

Computer experiments suggest that this is indeed true in the cases where the word length of some $w\in O_3$ is small enough. Note that the main ingredient of the proof in previous section is homology and degree theory which have direct analogues in case of $3$ dimension and higher. As such there is an obvious algorithm to proceed with in those cases. However, it turns out that it's hard to come up with $3-$cells corresponding to the $2-$cells $F$ and $J$ in $3$ dimension and unfortunately without those, all the relevant $2-$cycles are null homologous. Nonetheless, we can make some interesting observations which might prove helpful.
\begin{remark}
      Given $t_1t_2t_3\in O_3$, computer experiments suggest that it is always possible to choose $x_1,x_2,x_3,y_1,y_2,$ and $y_3$ in such a way that $t_1t_2t_3=x_1y_1x_2y_2x_3y_3$ and $x_1x_2x_3, y_1y_2y_3\in O_3$; i.e. it is possible to break the word $t_1t_2t_3$ into six subwords such that the two sets of three \emph{alternate} subwords together make two words of $O_3$. This dramatically cuts down the number of cells we need to consider for finding a relevant $2-$cycle.
\end{remark}
\begin{remark}
      Given $t_1t_2t_3\in O_3$, computer experiments suggest that it is always possible to choose $x_1,x_2,x_3,y_1,y_2,$ and $y_3$ in such a way that $t_1t_2t_3=x_1y_1x_2y_2x_3y_3$ and $x_1x_2x_3, y_1y_2y_3\in O_3$; i.e. it is possible to break the word $t_1t_2t_3$ into six subwords such that the two sets of three \emph{alternate} subwords together make two words of $O_3$. This dramatically cuts down the number of cells we need to consider for finding a relevant $2-$cycle.
\end{remark}

\end{document}